\documentclass[11pt]{article}
\usepackage{amssymb,amsmath, amsthm, setspace, color, graphics, fullpage}
\usepackage{mathpazo, mathptmx, flexisym}
\usepackage{fancyhdr}
\usepackage[authoryear]{natbib}
\usepackage{url}
\usepackage{breqn}
\usepackage[margin=1in]{geometry} 
\usepackage{setspace}             
\usepackage[all, cmtip]{xy}
\usepackage[english]{babel}
\usepackage{float}
\usepackage[section]{placeins}
\usepackage{mathtools}
\usepackage{caption}
\usepackage{subcaption}
\usepackage{tikz}
\usepackage{pdfpages}
\usetikzlibrary{arrows.meta,patterns,positioning}
\usepackage[pdfborder={0 0 0},final=true,colorlinks=true,linkcolor=magenta,citecolor=blue]{hyperref}
\usepackage[capitalize,nameinlink,noabbrev]{cleveref}
\numberwithin{equation}{section}
\newcommand{\tab}{\hspace*{2em}}              
\providecommand{\Macro}[2]{\expandafter\def\csname #1\endcsname{#2}}

\newcommand{\bm}[1]{\mathbf{#1}}                     
\newcommand{\diadt}[2]{t^{(#1)}_{#2}}                

\Macro{xj}{\bm{x}_{\diadt{n}{j}}}                              
\Macro{xjplus1}{\bm{x}_{\diadt{n}{{j+1}}}}                              
\Macro{xTj}{\bm{x}^T_{\diadt{n}{j}}}                           
\Macro{xTjplus1}{\bm{x}^T_{\diadt{n}{{j+1}}}}                           
\Macro{zj}{\bm{z}_{\diadt{n}{j}}}                              
\Macro{zjplus1}{\bm{z}_{\diadt{n}{{j+1}}}}                              
\Macro{zTj}{\bm{z}^T_{\diadt{n}{j}}}                           
\Macro{zTjplus1}{\bm{z}^T_{\diadt{n}{{j+1}}}}                           
\Macro{Mxj}{M_{\bm{x}_{\diadt{n}{j}}}}                         
\Macro{Mxjplus1}{M_{\bm{x}_{\diadt{n}{{j+1}}}}}                         
\Macro{Mzj}{M_{\bm{z}_{\diadt{n}{j}}}}                         
\Macro{Mzjplus1}{M_{\bm{z}_{\diadt{n}{{j+1}}}}}
\Macro{SxO}{S\times\Omega}                           
\Macro{SigmaSxO}{\mathfrak{S}\times\mathfrak{O}}     
\Macro{Lp}{L^p_{}(X,\mathfrak{B}(X),m)}              
\Macro{L2}{L^2_{}(X,\mathfrak{B}(X),m)}
\Macro{LinftyS}{L^{\infty}_{}(S \times \Omega,\mathfrak{S}\times\mathfrak{O},\mu)}
\Macro{LinftyS*}{L^{\infty}_{}(S \times \Omega,\mathfrak{S}\times\mathfrak{O},\mu)^*}
\Macro{Mvof}{\mathfrak{M}_{v \circ f}^{}}              
\Macro{ImPhi}{\text{Im}\;\phi}                         
\Macro{KerPhi}{\text{Ker}\;\phi}                       
\Macro{GQKerPhi}{G\backslash\text{Ker}\phi}            
\Macro{CoMap}{\alpha\circ\delta\circ\gamma}             
\Macro{GrpA}{(G,*)}                                    
\Macro{GrpO}{(G^\prime,\diamond)}                      
\theoremstyle{plain}

\newtheorem{cor}{Corollary}
\newtheorem{prop}{Proposition}
\newtheorem{lem}{Lemma}
\newtheorem{assumption}{Assumption}      

\theoremstyle{definition}
\newtheorem{defn}{Definition}[section]
\crefname{prop}{Proposition}{Propositions}
\Crefname{prop}{Proposition}{Propositions}
\crefname{lem}{Lemma}{Lemmas}
\Crefname{lem}{Lemma}{Lemmas}
\crefname{cor}{Corollary}{Corollaries}
\Crefname{cor}{Corollary}{Corollaries}

\theoremstyle{remark}
\newtheorem{rem}{Remark}[section]

\newcommand\reducefonten{\fontsize{10}{9}\selectfont}
\begin{document}
\singlespacing
\title{\textbf{Delegated Monitoring in Public--Private Credit Programs: Underinvestment, Overinvestment, and the Design of Subsidized Lending}}
\author{ G. Charles-Cadogan \thanks{University of Leicester, School of Business, Division of Accounting \& FinanceSchool ofAccounting \& Finance, School of Business, Leicester, LE2 1RQ; e-mail: \textcolor[rgb]{0.00,0.00,1.00}{\href{mailto:gocadog@gmail.com}{gocadog@gmail.com}}~\\~\\
	I thank John A. Cole for helpful comments. Any errors which may remain are my own.}
}
\date{\today}
\renewcommand\thefootnote{\fnsymbol{footnote}}
\maketitle
\renewcommand\thefootnote{\arabic{footnote}}
\thispagestyle{empty}
\begin{abstract}
\noindent This paper studies public--private partnerships that delegate access-to-credit programs to private equity and venture-capital intermediaries. The public sector seeks to relax credit rationing and expand lending to socially valuable firms, while delegated monitors screen applicants, allocate subsidized loans, and bear agency costs. The paper develops a mechanism-design model showing that the same delegated intermediation structure can generate both Stiglitz--Weiss underinvestment and De Meza--Webb overinvestment distortions. When screening is imperfect, interest-rate sorting may exclude creditworthy target firms. When subsidies weaken monitoring and repayment incentives, high-risk firms may obtain excessive credit. The operative distortion is determined by monitoring curvature and subsidy intensity. The model further predicts that the feasible spread between loan contracts narrows as expected risk increases. A sequential extension incorporates credit scoring and Bayesian updating, showing that welfare loss from misclassification is minimized when monitoring resources are allocated according to the marginal effect of borrower characteristics on posterior risk classification. Because granular borrower-level data are unavailable, the empirical component provides descriptive state-year evidence from SBA SBIC reports (2018--2025) and a calibrated simulation illustrating that the model's comparative statics are empirically recoverable. The empirical analysis is presented as design validation rather than as a causal test.
\end{abstract}

\newpage
\pagenumbering{arabic}
\renewcommand\thefootnote{\fnsymbol{footnote}}
\renewcommand\thefootnote{\arabic{footnote}}
\doublespacing
\section{Introduction}\label{sec:Intro}
\noindent Credit rationing remains a central motivation for public intervention in small-firm finance. In the canonical \citet{StiglitzWeiss1981} environment, asymmetric information prevents the interest rate from clearing the credit market because higher rates worsen borrower composition and borrower incentives. The policy response is often to use guarantees, subsidized loans, certification, or co-funding arrangements to expand access to credit. These instruments can be valuable when viable firms are discouraged from applying or are screened out by private lenders \citep{ColeSokolyk2016,CowlingLiuMinnetiZhang2016,FerrandoMulier2022}. Yet they also create a contracting problem: the public sector wants additional lending to socially valuable firms, while private delegated monitors retain their own risk, effort, and screening incentives.

\tab This paper studies that contracting problem in a public-private partnership (PPP) in which the public sector delegates loan allocation to private equity or venture-capital intermediaries. The delegated monitor observes imperfect signals of firm type, allocates subsidized credit, and can audit borrower declarations. This structure captures programs in which public funds or guarantees are combined with private screening and monitoring. The central result is a public-private partnership puzzle. Delegated monitoring can generate underinvestment when interest-rate sorting and imperfect screening exclude target firms; but subsidized credit can also generate overinvestment when high-risk firms obtain excessive subsidized funding. Benevolent public intent therefore does not by itself imply an efficient access-to-credit mechanism.

\tab The paper makes three contributions. First, it connects credit-rationing theory with delegated monitoring in PPP finance. Second, its novel theoretical contribution is to show that the same delegated public-private credit mechanism can embed both Stiglitz-Weiss underinvestment and De Meza-Webb overinvestment margins, with the operative distortion determined by monitoring curvature and subsidy intensity \citep{DeMazaWebb1987,DeMezaWebb1992}. Third, it extends the static mechanism to a sequential contract in which the delegated monitor updates firm type using credit-scoring information and realized performance. The sequential model clarifies how effort, auditing, and posterior classification probabilities affect the welfare loss from misclassification.

\subsection{Related literature}\label{subsec:RelatedLiterature}

\tab The paper is related first to the literature on credit rationing, discouraged borrowers, and SME access to finance. The classic credit-rationing mechanism of \citet{StiglitzWeiss1981} remains the benchmark for understanding why viable firms may not receive credit. Subsequent work documents that discouraged borrowers can be observationally important and that many such firms would receive credit if they applied \citep{ColeSokolyk2016,CowlingLiuMinnetiZhang2016,FerrandoMulier2022,CowlingLiuZhang2022}. These results motivate public intervention, but they do not by themselves determine whether the intervention should be direct public lending, guarantees, or delegation to private monitors.

\tab Second, the paper contributes to work on public credit guarantees and public risk-bearing. Credit guarantees are now among the most widely used policy instruments for SME finance \citep{BeckKlapperMen2010,OECD2017,WorldBankFIRST2015,CrawfordCuiKewley2024}. Existing work emphasizes additionality, fiscal sustainability, counterfactual evaluation, and the possibility that guarantees alter lender and borrower incentives \citep{Cowling2010,OnoUesugiYasuda2013,AnginerTorreIze2014,CarlettiLeonelloMarquez2023,MarshSharma2024}. The present paper complements that literature by modeling how delegated private screening interacts with public subsidy design.

\tab Third, the paper builds on delegated monitoring, venture-capital contracting, and credit-scoring literatures. Delegated monitoring is central to financial intermediation \citep{LelandPyle1977,Diamond1984,Chan1983,BesankoKanatas1993}, while venture capitalists supply monitoring, selection, and governance in opaque firms \citep{Sahlman1988,Sahlman1990,AdmatiPfliederer1994,GompersGornallKaplanStrebulaev2020,BernsteinGiroudTownsend2022}. The sequential component of the model is linked to credit scoring and information production, including the use of hard information, soft information, and digital footprints in lending \citep{FrameSrnivasanWoosley2001,DeYoungGlennonNigro2008,LibertiPetersen2021,BergBurgGombovicPuri2021,FusterGoldsmithPinkhamRamadoraiWalther2022}.

\tab Finally, the mechanism-design structure follows the public procurement and regulation tradition in \citet{LaffontTirole1993}. The distinctive feature here is that the principal is not simply procuring a good or regulating a firm. The public sector delegates access-to-credit decisions to a private intermediary whose screening incentives determine whether public funds reach target firms or are captured by non-target firms.

\tab In this paper, information is communicated to the delegated lender through borrower declarations and observable firm characteristics. The firm signals its type by completing an application, and the delegated lender uses decision rules to assign firms to risk categories. The declaration can be audited, but monitoring is costly. Agency costs therefore consist of monitoring expenditures, bonding requirements, and residual losses when borrowers strategically misrepresent type \citep{JensenMeckling1976}. The rest of the paper proceeds as follows. \autoref{sec:TheModel} introduces the model. \cref{sec:OptimalContractFirstBest} studies the symmetric-information benchmark, and \autoref{sec:OptimalContractSecondBest} studies asymmetric information. \cref{sec:SequentialContracts} extends the analysis to sequential contracts with credit scoring and posterior updating. \autoref{sec:Conclusion} concludes. Proofs are collected in the Appendix.

\section{The Model}\label{sec:TheModel}
This section develops a mechanism-design model of a public-private credit program. The public sector is risk neutral and delegates loan allocation to risk-averse venture capitalists or portfolio managers who screen risk-averse, liquidity-constrained firms. The model identifies when benevolent public objectives fail to implement truthful revelation and how subsidies interact with firm type, loan size, and monitoring costs.

The equilibrium analysis is intentionally partial. The paper characterizes incentive-compatible separating and semi-separating allocation conditions within a direct-revelation screening environment; it does not claim to characterize every perfect Bayesian equilibrium of an unrestricted dynamic game. Separating allocations are unique only under the maintained single-crossing, concavity, and convex agency-cost restrictions that make the relevant objective single-peaked. When those restrictions fail, the admissible allocation is set-valued, and the paper's claim is limited to feasibility and comparative statics rather than uniqueness. The transfer or subsidy is treated as a policy instrument rather than as the outcome of a political-choice problem. The welfare criterion is a normative social-surplus benchmark: it is politically implementable only to the extent that program statutes or agency rules commit the public sector to observable target weights, monitoring rules, and budget constraints. A full political-economy model of subsidy choice is therefore outside the scope of the paper.

The notation $\lambda$ is used for classification probabilities at three levels of the model. The distinction is summarized as follows:
\begin{table}[H]
\centering
\caption{Classification-probability notation}\label{tab:LambdaNotation}
\begin{tabular}{p{0.22\textwidth}p{0.28\textwidth}p{0.40\textwidth}}
\hline
Notation & Role & Interpretation \\
\hline
$\overline{\lambda},\underline{\lambda}$ & Static prior probabilities & The delegated lender's prior probability that an applicant is high-risk or low-risk before the sequential performance signal is observed. \\
$\lambda(\pmb{c})$ & Classification probability & The probability generated by a credit-scoring rule from borrower characteristics $\pmb{c}$. This object is endogenous to the classification technology specified by the lender. \\
$\lambda(\pmb{c}\mid s)$ & Posterior probability & The updated probability after observing a performance signal $s\in\{s_G,s_B\}$ and applying Bayes' rule. This is the probability relevant for misclassification and welfare loss. \\
\hline
\end{tabular}
\end{table}

Consider an economy with $N$ firms, each with the quasi-linear utility function
\begin{equation}
   U_i(x;\theta_i;y_i)=V_i(x,\theta_i)+y_i,\quad i=1,2,\dotsc ,N\label{eq:QuasiLinearUtilityModel}
\end{equation}
where $x$ is the type-dependent loan made available by the delegated manager, $\theta_i$ is firm $i$'s true type, and $y_i$ is a transfer generated by preferential contracting or certification. The function $V_i$ is concave in loan size, with $V_i^\prime>0$ and $V_i^{\prime\prime}<0$. From the public sector's perspective, firm surplus is the social benefit induced by the loan-transfer mechanism. Utility is separable in the loan $x$ and the transfer $y_i$, so the transfer required by the firm is independent of the loan amount. The authorities know $V_i$ but not $\theta_i$, and they make an amount $x^*$ available through low-interest or forgivable loans. A mechanism is therefore needed to allocate $x^*$ across firms. In what follows, ``government,'' ``delegated venture capitalists,'' and ``delegated portfolio managers'' refer to the public-private allocation arrangement unless a distinction is needed.
\singlespacing
\begin{defn}[Allocation Mechanism]\label{def:AllocationMechanism}
Let $X$ be the set of all possible loan allocations that the authorities may choose in the economy i.e. $X$ is a vector-subspace in $\mathbb{R}^N_+\backslash \{0\}$, so that
   \begin{equation}
      X=\{\pmb{x}|\;\pmb{x}=(x^1,x^2,\dotsc,x^N)\in \mathbb{R}^N_+\backslash \{0\}\};
   \end{equation}
   $\pmb{\Theta}$ be the set of all possible types of firms i.e.
   \begin{equation}
       \pmb{\Theta}=\{\theta_i|\;\theta_i\;\; \text{is $i$-th firm type},\;i=1,2,\dotsc,N\};
   \end{equation}
   $S_i$ be the set of strategies available for the $i$-th firm i.e.
   \begin{align}
      S_i&=\{\hat{\theta}_i|\hat{\theta}_i\gtreqqless\theta\}\label{eq:StrategySpace}\\
      \intertext{where $\hat{\theta}_i$ is a noisy signal of firm $i$'s true type, i.e.}
      \hat{\theta}_i&=\theta_i+\epsilon_i
   \end{align}
   for some unobserved factor $\epsilon_i$. The allocation mechanism is defined as a payoff function that maps elements of the strategy space into the set of loan allocations. That is,
   \begin{equation}
      y_i:S_1\times S_2\times\dotsc\times S_N\rightarrow X\label{eq:AllocationMechanismPayoffFunc}
   \end{equation}
\end{defn}
\begin{rem}
   See \cite{Myerson1989} and \cite{Ledyard1989} for more technical mechanism designs.\qed
\end{rem}
\doublespacing
\begin{center}
  \sc{\textbf{Public sector objectives}}
\end{center}
\noindent The government's objective is to choose an optimal allocation of loans through delegated investment companies while possessing incomplete information. Let $C(\pmb{x})$ be the cost of monitoring this program, and $\pmb{\theta}=(\theta_1,\theta_2,\dotsc,\theta_N)$ be a vector of types, so that
\begin{equation}
    \pmb{x}^*(\pmb{\theta})\in\arg \max_{\pmb{x}\in X} \biggl\{ \sum_{i=1}^N V_i(x^i,\;\theta_i)-C(\pmb{x}(\pmb{\theta}))\biggr\}\label{eq:OptimalAllocation}
\end{equation}
where the maximand is the government's social welfare function. In this setup, the government extracts aggregate surplus from the $N$ firms. Each firm reports a type $\hat{\theta}$ to the delegated investment company, for example through a loan application or questionnaire. The reported type may differ from the true type $\theta$. Based on reported information, the government chooses a transfer $y_i(\hat{\theta}_1,\hat{\theta}_2,\dotsc,\hat{\theta}_N)$ to offset externalities associated with firm size, ownership characteristics, or informational opacity. For instance, women-owned business enterprises may be locked out of networks that generate contract opportunities, while small disadvantaged businesses may face weak collateral or discriminatory lending frictions \citep{BatesBradford2008,JacksonBates2013,ColeSokolyk2016}.
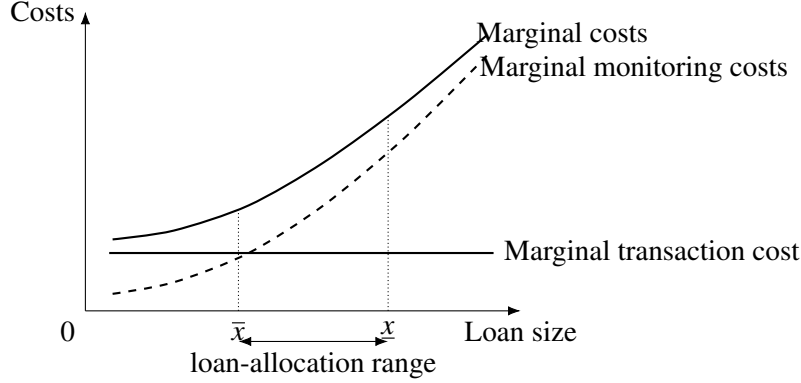
\begin{figure}[!htbp]
   \centering
   \captionof{figure}{Government's cost structure}
   \label{fig:PubPrivPrinCostStructure}
   \begin{tikzpicture}[scale=0.9,>=Latex]
      \draw[->] (0,0) -- (6.4,0) node[below] {Loan size};
      \draw[->] (0,0) -- (0,4.4) node[left] {Costs};
      \node[below left] at (0,0) {$0$};
      \draw[thick] (0.35,0.85) -- (6.0,0.85) node[right] {Marginal transaction cost};
      \draw[thick,dashed] plot[smooth] coordinates {(0.4,0.25) (1.3,0.42) (2.4,0.85) (3.5,1.55) (4.7,2.55) (5.9,3.75)};
      \node[right] at (5.65,3.55) {Marginal monitoring costs};
      \draw[thick] plot[smooth] coordinates {(0.4,1.05) (1.3,1.18) (2.4,1.55) (3.5,2.18) (4.7,3.05) (5.9,4.05)};
      \node[right] at (5.6,4.05) {Marginal costs};
      \draw[densely dotted] (2.25,0) -- (2.25,1.48);
      \draw[densely dotted] (4.45,0) -- (4.45,2.88);
      \node[below] at (2.25,0) {$\overline{x}$};
      \node[below] at (4.45,0) {$\underline{x}$};
      \draw[<->] (2.25,-0.45) -- (4.45,-0.45) node[midway,below] {loan-allocation range};
   \end{tikzpicture}
   \begin{minipage}{0.85\linewidth}\reducefonten
   	  \autoref{fig:PubPrivPrinCostStructure} gives the loan-allocation range a geometric interpretation. The interval $[\overline{x},\underline{x}]$ is interior to the cost curve and measures the feasible spread between the loan offered to a high-risk firm and the loan offered to a low-risk firm. Curvature matters because movement into steeper regions of the monitoring-cost schedule raises the marginal cost of expanding this interval. Hence higher monitoring intensity contracts the feasible range of loans: range contraction is not merely an algebraic implication of the model, but the geometric consequence of operating on a steeper cost schedule. This same geometry connects the static allocation problem to the strategic sections below. When cost curvature narrows the feasible interval, sorting becomes tighter, strategic misrepresentation becomes harder to accommodate without distorting the menu, and separating loan schedules become more demanding.
   	\end{minipage}
\end{figure}
\noindent The transfer payment is
\begin{equation}
    y_i(\hat{\theta}_1,\hat{\theta}_2,\dotsc,\hat{\theta}_N)=\omega_i+\sum_{j\neq i}V_j(x^j,\hat{\theta}_j)-C(\pmb{x}(\hat{\pmb{\theta}}))\label{eq:VickeryGrovesTransfer}
\end{equation}
where $\omega_i$  is firm $i$'s valuation of its surplus. The formula in \eqref{eq:VickeryGrovesTransfer} is designed so that each firm receives a transfer payment equal to a function of the other firms' reported surplus net of the cost of providing loans. This is a Vickrey-Groves type mechanism: by making the transfer depend on the reported surplus of firms other than $i$, it reduces firm $i$'s incentive to influence its subsidy by misrepresenting its own surplus. Assuming that $(y,g(x))$ is a contract with transfer payment $y$ and loan schedule $g(x)$, $\omega_i$ can be interpreted as firm $i$'s minimum valuation for accepting the contract:
\begin{equation}
     y_i(\hat{\pmb{\theta}})-\omega_i
     +\sum_{j\neq i}V_j(x^{*j},\hat{\theta}_j)
     -C(\pmb{x}(\hat{\pmb{\theta}}))
     \geq 0.\label{eq:FirmValuationOfContract}
\end{equation}
If there is no budget balance (i.e. the total surplus extracted by the government does not equal total costs) then a dominant strategy for the $i$-th firm would be to report $\max\{\omega_i,\;0\}$ . Under this strategy, the net payoff or transfer payment for the firm is $\max\{y_i(\hat{\pmb{\theta}})-\omega_i,\;0\}$. With this mechanism, the firm will be truthful in a Nash equilibrium in which other firms are truthful. However, this is a weakly dominant strategy since it also holds in situations where other firms do not tell the truth but the net payoff is still greater than zero.\footnote{ See \cite[pp.~173-174]{Rasmusen1989}; \cite[p.~268]{FudenbergTirole1991}}
\subsection{Incentive compatibility and rationality constraints}
In order to identify inefficiency in the transfer payment structure, we simplify the model in \eqref{eq:QuasiLinearUtilityModel} and \eqref{eq:VickeryGrovesTransfer} by considering a representative agent (i.e. firm). Further, we impose more structure on the quasi-linear utility function for expository purposes. All subscripts are removed for notational convenience so that the firm receives utility
\begin{equation}\label{eq:QuasiLinearUtilityModelSeparateType}
 	U(x,\theta; y)=\theta V(x)+y
\end{equation}
where $x$ is the loan amount for which the firm applies. $\theta V(x)$ is the surplus, separable in type $\theta$, the firm would derive from the loan. The government offers $y$ contract dollars thus specifying how much the firm gets if it applies for loan $x$. The firm either accepts or rejects the government offer according to whether it exceeds the firm's reservation value of its surplus. Suppose that there are two types of firms: high-risk firms (type-$\overline{\theta}$ ) and low-risk firms (type-$\underline{\theta}$ ).  If the government knew the true value of $\theta$  with certainty it would offer a fixed loan and require the firm to pay (i.e. report) an amount, $-\theta V(x)$, in order to get the loan. In that way, it extracts any surplus which the firm may get from the loan. Thus, $(\underline{y},\underline{x})$   and $(\overline{y},\overline{x})$  denote a pair of contracts for type-$\underline{\theta}$  and type-$\overline{\theta}$  firms respectively. Because delegated investment companies are risk-averse venture capitalists they offer smaller loans to high-risk firms and larger loans to low-risk firms so that $\overline{x}\leq \underline{x}$.\footnote{This is notationally cumbersome but consistent with the notion that high-risk firms receive smaller loans than low-risk firms. One could argue that this policy is contrary to the goals of a subsidy program designed to help ``at risk" firms. However, the underlying motivation stems from the fact that public-private sector resources are scarce so they seek to ration subsidies where they would do the most good. That is, subsidizing high-risk firms could prove to be a very expensive waste of resources since those firms would tend to go bankrupt with greater frequency than low-risk firms.} Their problem is one of maximizing the amount of surplus extracted from firms that apply for loans. Let $\overline{\lambda}$  be the lender's subjective probability for a high-risk firm and $\underline{\lambda}$  the corresponding probability for a low-risk firm. Technically, lenders may formulate these probabilities from experience via some classification procedure that depends on firm characteristics. For example, they may use a Bayesian discriminant analysis procedure.\footnote{Refer to \cite{Lo1986}; \citep[Ch.~6]{Anderson2003}, \cite[Chapter~30]{BrealeyMyers1991} and the references therein} In any case, all that is required is that
\begin{equation}
    \overline{\lambda}+\underline{\lambda}=1
\end{equation}

\tab In order for the mechanism to be effective firms must want to participate. Two sets of criteria are required for this to hold. One set focuses on the firm's individual rationality (IR1) - the firm would apply for a loan only if it has a net positive utility/valuation for the loan. The other set (IR2) focuses on the compatibility of the incentives provided by the government with the firm's choices. That is, the firm must choose the loan-transfer bundle provided for its risk type. These criteria are represented below.
\begin{enumerate}
   \item[IR1:]
      \vspace{-2ex}\begin{equation}
         \underline{\theta} V(\underline{x}) + \underline{y} \geq 0
      \end{equation}
   \item[IR2:]
      \vspace{-2ex}\begin{equation}
         \overline{\theta} V(\overline{x}) + \overline{y} \geq 0
      \end{equation}
\end{enumerate}
IR1 and IR2 imply that the net utility which the firm derives from loan x and subsidy y must leave it at least as well off as before.  That is, net utility must at least be equal to the firm's reservation utility, which is zero here. The ``bar" represents the loan and subsidy amounts procured by the low/high type firm, accordingly. The truth telling or incentive compatibility (IC) constraints for the firm are given by
\begin{enumerate}
   \item[IC1:]
      \vspace{-2ex}\begin{equation}
         \underline{\theta} V(\underline{x}) + \underline{y} \geq \underline{\theta}V(\overline{x}) + \overline{y}
      \end{equation}
   \item[IC2:]
      \vspace{-2ex}\begin{equation}
         \overline{\theta} V(\overline{x}) + \overline{y} \geq \overline{\theta}V(\underline{x}) + \underline{y}
      \end{equation}
\end{enumerate}
These constraints compel a utility-maximizing firms to select the bundle designated for its type. For example, IC(2) suggests that a high-risk firm should get at least as much utility from the bundle designated for its type compared to what it would get if it chose the bundle designated for the low-risk type and vice versa. Thus, firms get less utility when they do not tell the truth.
\section{Optimal Contract With Symmetric Information (First Best)}\label{sec:OptimalContractFirstBest}
Under the premise of symmetric information, delegated monitoring may not be necessary because information is fully revealed to public sector officials. Thus, the firm reports type $\theta$  truthfully and the authorities extract the appropriate amount of surplus $\theta V(x)$ while adjusting for the cost $C(x)$ of providing the loan. The optimal loan in that case satisfies:
\begin{align}
   x^*&\in\arg\max_x\biggl\{\theta V(x) - C(x)\biggr\}\\
   \intertext{where the necessary first-order condition is given by}
   \theta V^\prime(x)&=C^\prime(x)\\
   \intertext{The second-order condition for maxima is satisfied by}
   \theta V^{\prime\prime}(x)&-C^{\prime\prime}(x)<0
   \intertext{due to concavity of $V$ and convexity of $C$. The optimal contract under this setup is}
   &\{y(\theta),x^*(\theta)\}.
\end{align}
\section{Optimal Contract With Asymmetric Information (Second Best)}\label{sec:OptimalContractSecondBest}
In this case the authorities do not know $\theta$, the firm's type. However, they distribute transfers $\underline{y},\;\overline{y}$ on the basis of the low-risk and high-risk firms' noisy signal $\hat{\underline{\theta}},\;\hat{\overline{\theta}}$, respectively. Furthermore, they team up with venture capitalists who are more informed about the loan allocation process $(x)$. Venture capitalists offer two types of contracts $(\underline{y},\underline{x})$  and $(\overline{y},\overline{x})$   to low-risk and high-risk firms, respectively. Firms believe that their chances of procuring a loan is improved if they have a government contract in hand. So sequentially, transfer $(y)$ may represent a procured government contract which occurs before loan $x$ is obtained.\footnote{Arguably, this leads to a situation in which we have $x(y(\theta))$. That is, the amount of loans is a functional of the amount of transfers} This premise is reasonable. For example, \citet[Prop.~6]{DeMazaWebb1987} contemplates interest income subsidy for banks in order to attain socially desirable equilibrium in their model. Transfer $y$ is a dual to interest rate subsidy since it reduces firm risk.

The delegated monitor's objective is written as expected extracted surplus rather than as a full accounting profit function. This is a reduced-form way of representing the private intermediary's objective after the public transfer rule has been fixed. Any direct subsidy, guarantee, or procurement contract that is paid according to the policy rule enters the firm's participation constraint or reduces effective repayment risk; the delegated monitor then chooses the loan menu to maximize the expected surplus recoverable from the applicant net of monitoring costs. Equivalently, one can interpret the objective as expected profit net of delegated-monitoring cost after normalizing fixed program payments and treating the public subsidy rule as predetermined. The formulation therefore isolates the screening margin; it does not assume that the intermediary socially internalizes the subsidy cost.

Venture capitalists maximize the expected value of the surplus extracted from a given firm when it applies for a loan, subject to costs, rationality, and incentive constraints as follows:
\begin{center}
   \sc{Venture Capitalist Problem}
\end{center}
\vspace{-4ex}\begin{align}
   \max_{\underline{x},\;\overline{x}}&\biggl\{\overline{\lambda}(\theta\;V(\overline{x})-C(\overline{x})+\underline{\lambda}(\theta V(\underline{x})-C(\underline{x}))\label{eq:VentureCapitalProblem}
   \biggr\}\\
   &\text{Subject to IR1 IR2; IC1 IC2}
\end{align}
In order to maintain the benevolent aspect to the transfer scheme high-risk firms first must be compensated with a given amount of reservation or baseline utility
\begin{equation}
    \overline{U}\geq 0\label{eq:ReserveUtilityHighRiskFirm}
\end{equation}
in order to induce their participation.  In the event IR2 is binding we have
\begin{align}
   \overline{\theta}\,V(\overline{x})+\overline{y}&=\overline{U}\label{eq:HighRiskReservationU}
\end{align}
Low risk firms are indifferent to the benevolent policy when IC(1) is binding so that
\begin{align}
   \underline{\theta} V(\underline{x})+\underline{y}&=\underline{\theta}V(\overline{x})+\overline{y}\\
   \intertext{After substitution for $\overline{y}$ from \eqref{eq:HighRiskReservationU}, we get}
   \underline{\theta}V(\underline{x})+\underline{y}&=\overline{U}+(\underline{\theta}-\overline{\theta})V(\overline{x})
\end{align}
Hence the total utility derived by the low-risk firms from this program exceeds that provided for high-risk firms by the informational rent
\begin{equation}
   (\underline{\theta}-\overline{\theta})V(\overline{x}),\quad\underline{\theta}>\overline{\theta}.\label{eq:SortingCondition}
\end{equation}
The inequality $\underline{\theta}>\overline{\theta}$ is a maintained ordering of types, not by itself a sorting equilibrium condition. The economically relevant restriction is instead a participation-composition condition: the contract menu must induce a participant pool in which the mass of low-risk firms is large enough for the program's social-surplus objective to dominate the rents and monitoring costs generated by high-risk participation. Let $m_{\underline{\theta}}$ and $m_{\overline{\theta}}$ denote the equilibrium measures of participating low-risk and high-risk firms under the offered menu.
\begin{lem}[Participation-composition condition]\label{lem:BenevolentSortinngCondition}
   Suppose $\underline{\theta}>\overline{\theta}$ and the low-risk rent in \eqref{eq:SortingCondition} is nonnegative. A necessary participation-composition condition for the public-private partnership to implement its benevolent objective is
   \begin{equation*}
      m_{\underline{\theta}}\geq m_{\overline{\theta}},
   \end{equation*}
   with strict inequality whenever high-risk participation raises monitoring costs or subsidy costs at the margin. Equivalently, the program is benevolent only if the separating menu attracts enough low-risk participation to offset the fiscal and monitoring costs induced by high-risk participation.\qed
\end{lem}
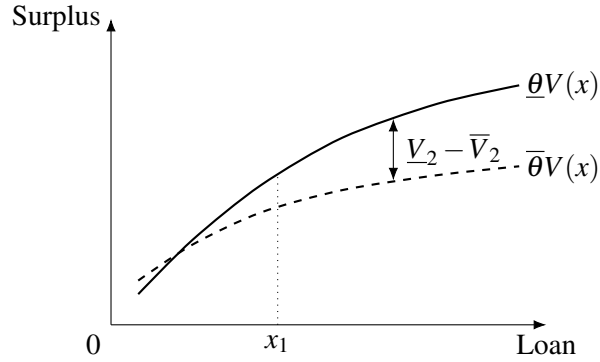
\begin{figure}[!htbp]
   \centering
   \captionof{figure}{Loan rent to low-risk firms}
   \label{fig:PubPrivSingleCrossProp}
   \begin{tikzpicture}[scale=0.9,>=Latex]
      \draw[->] (0,0) -- (6.4,0) node[below] {Loan};
      \draw[->] (0,0) -- (0,4.5) node[left] {Surplus};
      \node[below left] at (0,0) {$0$};
      \draw[thick] plot[smooth] coordinates {(0.4,0.45) (1.2,1.25) (2.2,2.05) (3.4,2.75) (4.8,3.25) (6.0,3.52)};
      \node[right] at (5.95,3.52) {$\underline{\theta}V(x)$};
      \draw[thick,dashed] plot[smooth] coordinates {(0.4,0.65) (1.2,1.18) (2.2,1.65) (3.4,1.98) (4.8,2.20) (6.0,2.33)};
      \node[right] at (5.95,2.33) {$\overline{\theta}V(x)$};
      \draw[dotted] (2.45,0) -- (2.45,2.22);
      \node[below] at (2.45,0) {$x_1$};
      \draw[<->] (4.15,2.12) -- (4.15,3.03);
      \node[right] at (4.2,2.58) {$\underline{V}_2-\overline{V}_2$};
   \end{tikzpicture}
\end{figure}
\FloatBarrier\noindent The implications are twofold. First, benevolence is a statement about the induced composition of participants, not merely about the ordering of risk types. A menu that attracts mostly high-risk firms may be privately attractive to applicants yet inconsistent with the public objective once monitoring and subsidy costs are included. Second, the wider the gap between low-risk and high-risk firm productivity, the larger the informational rent that accrues to low-risk firms. In \autoref{fig:PubPrivSingleCrossProp} this rent is depicted by $\underline{V}_2-\bar{V}_2$. The single crossing property \citep[p.~31]{Salanie2005} implies that loans below $x_1$, i.e. ``small loans", induce adverse selection since they attract high-risk firms who derive a surplus relative to low-risk firms. Thus, loan sizes above $x_1$ are more likely to support a participant pool with sufficient low-risk participation. The following proposition provides a robust first-order condition to the solution in \eqref{eq:VentureCapitalProblem} extended over a continuum of types when individual rationality constraints are binding.
\begin{prop}[First order condition for optimality]\label{prop:LoanFirstOrderCondition}
Assume that the following ``regularity" conditions hold:
\begin{enumerate}
    \item[i.]   the allocation of loans is uniformly distributed over the interval $[\overline{x}, \underline{x}]$
    \item[ii.]  the distribution of firm types is strictly monotonic.
    \item[iii.] the marginal change in loan portfolio is equal to the incremental loan made to a type-$\theta$  firm with probability $\lambda$.
\end{enumerate}
Let $F(x),\;\;G(\theta)$  be the distribution function for loans and firm type respectively. Then the delegated monitor implicitly solves the following continuous analog to \eqref{eq:VentureCapitalProblem} over a continuum of types in order to determine the ``optimal" loan allocation:
\begin{align}
   x^*\in&\arg\max_{x}\int^{\underline{x}}_{\overline{x}}\int^{\underline{\theta}}_{\overline{\theta}}\big[\theta V(x)-\lambda(\theta)C(x)\big]dF(x)dG(\theta)\label{eq:SecondBestOptimaIntegral}\\
   \intertext{and the generalized form of the first-order condition is}
   &\theta V_{\tilde{x}}-\lambda C^\prime(\tilde{x})= 0\label{eq:MUcostFOC}
\end{align}
for some $\tilde{x}\in [\overline{x}, \underline{x}]$ and optimal loan set
\begin{equation}\label{eq:OptimalLoanSet}
   \mathfrak{L}=\biggl\{x\biggl |\;\theta^*V_x-\lambda(\theta^*)C_x=0 \;\; \text{and}\;\;\theta^*=\arg\max_\theta g(\theta)\biggr\}\bigcap\biggl\{[\overline{x},\underline{x}]\biggr\}
\end{equation}
\end{prop}
\begin{cor}[Marginal utility and marginal cost of loan]\label{cor:LoanOptimalityMC_Condition}
   For optimality, the marginal utility of a loan equals the expected marginal cost of providing that loan.
\end{cor}
\subsection{Semi-Separating Equilibrium for Second Best}
In this setup there are two components of lender costs associated with providing a loan (other than the amount of the loan itself). One is transaction cost. This should be fairly constant for all firms within given bounds. That is, attorney costs, clerical costs and other processing costs are fairly invariant to the loan amount provided in a given range. However, the second component, monitoring costs, depend on whether the firm is a high-risk or low-risk firm. It is reasonable to assume that high-risk firms have higher monitoring costs. For example, such firms may be subject to more frequent audits by government officials.\footnote{See \cite{BaronBesanko1984} and \cite{MookherjeePng1989}} Under this scenario, low-risk firms get higher loans and designated lenders incur lower monitoring costs that decrease as loan value increases. Since lenders lump all costs together, we have
\begin{assumption}
   For planning purposes lenders assign the costs incurred from the weighted average loan to all loans.
\end{assumption}
In that case the cost function is given by
\begin{equation}
   C(\mu_x),\qquad\mu_x=\overline{\lambda}\overline{x}+\underline{\lambda}\underline{x}
\end{equation}
where $\mu_x$  is the weighted average loan.\footnote{Government procurement contracts typically have an incentive scheme that is a mix between a fixed price contract and a ``cost plus" contract. Thus, government typically employs some kind of weighting procedure to estimate costs. Hence the cost structure assumed here is not far fetched. See \cite{Weitzman1980} for more on these incentive contract schemes.} By virtue of the convexity of the cost curve, we have
\begin{equation}
   \frac{\partial}{\partial x}C(x){\Bigl |}_{x=\overline{x}}\leq\frac{\partial}{\partial x}C(x){\Bigl |}_{x=\mu_x}\leq\frac{\partial}{\partial x}C(x){\Bigl |}_{x=\underline{x}}
\end{equation}
In the context of the necessary first-order condition in \cref{prop:LoanFirstOrderCondition} this implies
\begin{description}
   \item[Low risk:] \begin{subequations}\label{subeq:LowRiskMUMC}\vspace{-2ex}\begin{align}
                                                   \text{MU}_{\underline{x}}&=\underline{\theta}V_{\underline{x}} = \frac{\partial \mu_x}{\partial \underline{x}}C_{\mu_x}= \underline{\lambda}C_{\mu_x}\label{eq:MU_LowRiskFirm2}
                                                \end{align}
                                                \end{subequations}
   \item[High risk:] \begin{subequations}\label{subeq:HighRiskMUMC}\vspace{-4ex}\begin{align}
                                                  \text{MU}_{\overline{x}} &= \overline{\theta}V_{\overline{x}} = \frac{\partial \mu_x}{\partial \overline{x}}C_{\mu_x}= \overline{\lambda}C_{\mu_x}\label{eq:MU_HighRiskFirm2}
                                                \end{align}
                                                \end{subequations}
\end{description}
where by abuse of notation the subscripts on cost and utility functions imply derivative of those functions with respect to subscript. The marginal utility of a loan for low-risk firms $\text{MU}_{\underline{x}}$ is lower than the marginal cost of the average loan. That means that we can \emph{increase the amount of loans to low-risk firms} until it reaches some marginal cost threshold
\begin{equation}
   \text{MU}_{\underline{x}}\uparrow C^\star_{\mu_x}\label{eq:MU_LowRiskMCx}
\end{equation}
By contrast, the marginal utility of a loan for a high-risk firm $\text{MU}_{\overline{x}}$ is greater than the marginal cost of the average loan. Thus, we need to \emph{decrease the amount of loans to high-risk firms} until
\begin{equation}
   \text{MU}_{\overline{x}}\downarrow C^\star_{\mu_x}\label{eq:MU_HighRiskMCx}
\end{equation}
However,  \eqref{eq:MU_LowRiskMCx} and \eqref{eq:MU_HighRiskMCx} holds jointly only if the marginal utility of a loan is the same for each type of firm. In the context of \eqref{eq:MUcostFOC} in \cref{prop:LoanFirstOrderCondition}, this implies
\begin{equation}
   \underline{\theta}V_{\underline{x}}=\overline{\theta}V_{\overline{x}}\label{eq:MULRequalMUHR}
\end{equation}
Note that in \autoref{fig:PubPrivSingleCrossProp} the relation in \eqref{eq:MULRequalMUHR} implies $\overline{x}<x_1<\underline{x}$ since the marginal utility of a loan is greater for low-risk firms for loan size above $x_1$, i.e.  $\underline{\theta}V_{\underline{x}}>\overline{\theta}V_{\overline{x}}$.  More on point, from the lender's perspective a loan is an asset for which interest payment is received. So the marginal revenue  generated by interest rate $i$ for loan amount $x$ is $i\,x$. Let $\overline{i}$ and $\underline{i}$ be the rate of interest for loans given to high-risk and low-risk firms respectively.  In order for the lender to maximize profits we must have marginal revenue (MR) equal marginal cost (MC). That is, after substitution for MR in \eqref{eq:MULRequalMUHR}
\begin{align}
   \text{MR}&=\text{MC},\; \quad\text{implies}\;\overline{i}\overline{x}=\underline{i}\underline{x}.\;\;\text{But}\;\;\overline{x}<\underline{x}\Rightarrow \overline{i} > \underline{i}\label{subeq:FirmSort2}
\end{align}
\noindent That is, the lender maximizes profits by charging high-risk firms a higher interest rate than low-risk firms. This interest rate sorting condition result was anticipated by \citep[Thm.~4,~p.~397]{StiglitzWeiss1981}. However, in their model credit is denied to firms that are ``observationally indistinguishable from those who receive loans". We summarize this in
\begin{prop}[Interest rate sorting]\label{prop:TwoIntEquil}
  Venture capitalists (VCs) maximize profits according to \eqref{subeq:FirmSort2} in a two interest rates equilibrium that is functionally equivalent to \citep[Thm.~6,~p.~398]{StiglitzWeiss1981}. Specifically, VCs offer two contracts $(\overline{i},\,\overline{x};\,\overline{y})$ and $(\underline{i},\,\underline{x};\,\underline{y})$ respectively for firms with \emph{de facto} loan subsidies $y$ as indicated.\qed
\end{prop}
\noindent To avoid the \citet[Thm.~4]{StiglitzWeiss1981} conundrum, monotonicity in returns on loans provided by the partnership can be achieved by offering loan subsidies to high-risk firms based on the following arguments. For interest rate $i<\underline{i}$ low-risk and high-risk firms apply for loans. However, if $\overline{i}>i>\underline{i}$ only high-risk firms apply. Assuming linear transfers, suppose high-risk firms get a loan subsidy $\bar{y}$. We find that
\begin{align}
  \overline{i}(\bar{x}+\bar{y})=\underline{i}\underline{y}.\quad\text{If $\underline{i}=\bar{i}$,\;\;then}\;\;\bar{y}=\underline{x}-\bar{x}
\end{align}
So the partnership works by either offering separate contracts to high-risk and low-risk firms as in \cref{prop:TwoIntEquil} or by providing low-risk firms with a loan subsidy that tops up the high-risk firm to what a low-risk firm would get. The latter scenario is socially inefficient because it can induce moral hazard among high-risk firms. Low risk firms may also have an incentive to palm themselves off as high-risk firms in order to obtain subsidies. This is precisely the sense in which a policy designed to ``level the interest rate playing field" may create an overinvestment margin. If loan subsidies make high-risk firms privately more attractive, then VCs may oversubscribe to the program even when the induced allocation is not socially efficient. \citet{DeMazaWebb1987} arrived at this conclusion using a different approach.
\begin{prop}[Overinvestment]\label{prop:OverinvestmentProblem}
  If the government provides loan subsidies to high-risk firms in order to level the playing field for loans, then it induces an overinvestment problem.\qed
\end{prop}
\subsection{Loan monitoring}\label{subsec:LoanMonitor}
The level of monitoring depends on the location on the cost curve. For example, when marginal cost is decreasing there is low monitoring characterized by \eqref{eq:LowMonitor} below. High monitoring occurs when marginal cost is increasing in \eqref{eq:HighMonitor}
\begin{align}
   dC^\prime(C(\widehat{x})) = C^\prime(\widetilde{x})d\widetilde{x}+\bar{\lambda}\dfrac{\partial C^\prime(\widetilde{x})}{\partial \widetilde{x}}d\widetilde{x}=\overbrace{\underbrace{C^\prime(\widetilde{x})}_{(-)}d\bar{\lambda}+\bar{\lambda}\underbrace{C^{\prime\prime}(\widetilde{x})}_{(-)}\underbrace{(\overline{x}-\underline{x})}_{(-)}d\bar{\lambda}}^{\text{Low Monitoring}}<0\label{eq:LowMonitor}\\
   dC^\prime(C(\widehat{x})) = C^\prime(\widetilde{x})d\widetilde{x}+\bar{\lambda}\dfrac{\partial C^\prime(\widetilde{x})}{\partial \widetilde{x}}d\widetilde{x}=\overbrace{\underbrace{C^\prime(\widetilde{x})}_{(+)}d\bar{\lambda}+\bar{\lambda}\underbrace{C^{\prime\prime}(\widetilde{x})}_{(+)}\underbrace{(\overline{x}-\underline{x})}_{(-)}d\bar{\lambda}}^{\text{High Monitoring}}\ge 0\label{eq:HighMonitor}
\end{align}
\subsubsection*{Risky loans}
\tab The inequality in \eqref{eq:LowMonitor} shows that the expected loan received by high-risk firms increases with low monitoring when marginal costs are decreasing. That is, with low monitoring as long as marginal costs are decreasing the lower end of the range will continue to move to the right in \autoref{fig:PubPrivPrinCostStructure}. However, a more likely scenario is that delegated lenders operate in the region where marginal costs are increasing. In that case, for strategic misrepresentation to be supported under the conditions in \eqref{eq:HighMonitor} (i.e. when the relation is positive) we must have:
\begin{align}
   C^{\prime\prime}(\widetilde{x})\leq \dfrac{C^\prime(\widetilde{x})}{\overline{\lambda}(\underline{x}-\overline{x})}\label{eq:SecondOrderCostEffects}
\end{align}
Thus, second-order cost effects must be bounded by the inverse of the expected range of loans provided to high-risk firms in a non-steep region of the cost curve when marginal costs are increasing. This result is more intuitive than the result in \eqref{eq:LowMonitor} which assumes decreasing marginal costs (i.e. $C^\prime(\cdot) < 0$). Another way to look at this relation is as follows:
\begin{align}
   (\underline{x}-\overline{x})\leq \Big(\overline{\lambda}\dfrac{C^{\prime\prime}(\widetilde{x})}{C^\prime(\widetilde{x})}\Big)^{-1}\label{eq:RangeOfLoans}
\end{align}
Here \eqref{eq:RangeOfLoans} captures the intuitive aspects of the loan or preferred-contract decision. For instance, it suggests that the range of loans provided is inversely proportional to the principal's assessment of the firm's risk type $(\overline{\lambda})$ and the riskiness associated with loans given to that firm as captured by the cost effects $\dfrac{C^{\prime\prime}(\widetilde{x})}{C^\prime(\widetilde{x})}$. Further, the relationship clearly shows that the range of loans in which strategic misrepresentation can be profitable depends on the lender's subjective probabilities about $(\overline{\lambda})$ and the curvature of the cost curve.\footnote{This scenario resembles the Arrow-Pratt risk aversion concept. If we take the range of loans to be analogous to the government's risk premium, then as the cost curve gets steeper, i.e. the government becomes more risk-averse, the range of loans it is willing to provide decreases.} With high cost-curve curvature and a high probability of being high-risk, the range of loans is narrow; conversely, low curvature and a low probability of being high-risk widen the range. In the event costs are decreasing, the range is negative, i.e. the loan size for high-risk firms exceeds that for low-risk firms. These events lead to the following proposition:
\begin{prop}\label{prop:StrategicRange}
 Suppose a benevolent principal has subjective probability $(\overline{\lambda})$  for high-risk types ($(\underline{\lambda})$  for low-risk types) and offers loan  $\overline{x}$  ($\underline{x}$  for low-risk type) such that $\overline{x}\in[0,\underline{x})$. If $\widetilde{x}=\overline{\lambda}\overline{x}+(1-\overline{\lambda})\underline{x}$  is the average loan allocation and $C(\cdot)\in C^2[0,\infty)$  is convex and continuous agency costs, then the range of loan allocations in which agents can engage in strategic misrepresentation is as indicated in \eqref{eq:RangeOfLoans} independent of firm preferences.
\end{prop}
\noindent Intuitively, if low-risk firms solve a profit maximization problem subject to convex budget constraints, then the lender solves a dual problem of minimizing costs of lending to the firm. There exists a feasible, if not unique solution, which lies on a separating hyperplane \citep[e.g.,][]{LuenbergerYe2008}. In contrast, if high-risk firms have \citet{FriedmanSavage1948} type preferences, then according to \citet{Markowitz1952} these firms may be risk seeking over small loans, and loans above some intermediate threshold. That is, high-risk firms may be characterized by nonconvex preferences. Even if that were not the case, high-risk firms may require larger loans than low-risk firms. So their feasible regions may exceed that of low-risk firms and violate the separating hyperplane theorem. This scenario is depicted in \cref{fig:PubPrivPartnerHypePlane} and \cref{fig:PubPrivPartnerHypePlaneViolate}.
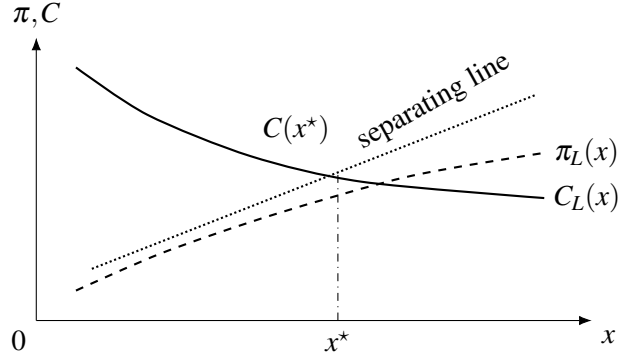
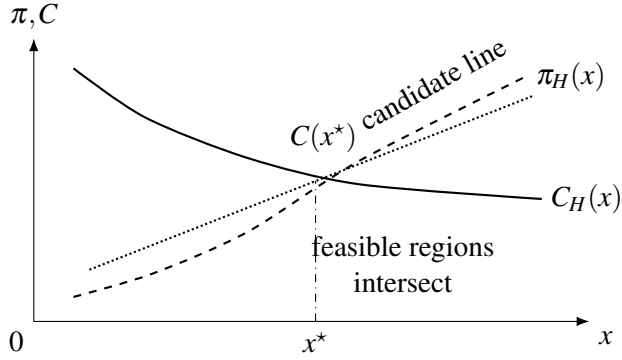
\begin{figure}[htb!]
   \centering
   \begin{subfigure}[t]{0.92\textwidth}
      \centering
      \begin{tikzpicture}[x=1.05cm,y=0.72cm,>=Latex]
         \draw[->] (0,0) -- (7.0,0) node[below right] {$x$};
         \draw[->] (0,0) -- (0,5.2) node[above] {$\pi,C$};
         \node[below left] at (0,0) {$0$};
         \draw[thick] plot[smooth] coordinates {(0.5,4.65) (1.4,3.78) (2.7,3.02) (4.1,2.55) (6.4,2.25)};
         \draw[thick,dashed] plot[smooth] coordinates {(0.5,0.55) (1.6,1.25) (3.0,1.95) (4.7,2.65) (6.4,3.08)};
         \draw[thick,densely dotted] (0.7,0.95) -- (6.3,4.15);
         \draw[dash dot] (3.8,0) -- (3.8,2.70);
         \node[below] at (3.8,0) {$x^\star$};
         \node[right] at (2.735,3.50) {$C(x^\star)$};
         \node[right] at (6.4,2.25) {$C_L(x)$};
         \node[right] at (6.4,3.08) {$\pi_L(x)$};
         \node[rotate=29,above] at (5.15,3.65) {separating line};
      \end{tikzpicture}
      \caption{Convex low-risk case: the separating line supports a feasible allocation at $x^\star$.}
      \label{fig:PubPrivPartnerHypePlane}
   \end{subfigure}

   \vspace{1.5ex}

   \begin{subfigure}[t]{0.92\textwidth}
      \centering
      \begin{tikzpicture}[x=1.05cm,y=0.72cm,>=Latex]
         \draw[->] (0,0) -- (7.0,0) node[below right] {$x$};
         \draw[->] (0,0) -- (0,5.2) node[above] {$\pi,C$};
         \node[below left] at (0,0) {$0$};
         \draw[thick] plot[smooth] coordinates {(0.5,4.65) (1.4,3.75) (2.7,3.00) (4.1,2.52) (6.4,2.25)};
         \draw[thick,dashed] plot[smooth] coordinates {(0.5,0.45) (1.5,0.88) (2.7,1.65) (4.1,2.95) (6.2,4.50)};
         \draw[thick,densely dotted] (0.7,0.95) -- (6.3,4.15);
         \draw[dash dot] (3.55,0) -- (3.55,2.58);
         \node[below] at (3.55,0) {$x^\star$};
         \node[right] at (3.11,3.38) {$C(x^\star)$};
         \node[right] at (6.4,2.25) {$C_H(x)$};
         \node[right] at (6.2,4.50) {$\pi_H(x)$};
         \node[rotate=29,above] at (5.15,4.05) {candidate line};
         \node[align=center] at (4.65,1.05) {feasible regions\\intersect};
      \end{tikzpicture}
      \caption{Nonconvex high-risk case: the payoff schedule intersects the candidate separation.}
      \label{fig:PubPrivPartnerHypePlaneViolate}
   \end{subfigure}
   \caption{Separating geometry for lender and firm feasible regions. Solid curves denote lender cost schedules, dashed curves denote firm payoff schedules, and dotted lines denote candidate separating hyperplanes. The regular convex case admits a supporting separating line; the high-risk case violates the separating geometry because the feasible regions intersect.}
\end{figure}
\FloatBarrier
\section{Sequential contracts}\label{sec:SequentialContracts}
This section extends the analysis to long-term contracts in which venture capitalists provide monitoring as part of the loan-allocation mechanism. We assume that contracts are spread over two periods indexed by dates 0, 1. At date 0 the firm applies for a loan whereupon it is classified either as high-risk or low-risk. At date 1 the firm's performance is reviewed and it is reclassified accordingly. There are three possibilities: the firm experiences a bad state, good state or becomes bankrupt.
\subsection*{Incentives to apply}
\tab All contracts are assumed to be complete \citep[Ch.~6]{Salanie2005}. That is, we assume that the venture capitalist is committed to the length of the contract and that no renegotiation is allowed. However, if the firm goes bankrupt, venture capitalists receive its salvage value. At date 0 nature selects a firm and the venture capitalist assigns an a priori type probability $\lambda(\pmb{c})$ to the firm after it applies for a loan, where
\begin{align}
   \pmb{c}=[c_{i1},\ldots,c_{im}]^T
\end{align}
is a vector of firm and borrower characteristics obtained when the loan application is filled out.\footnote{Several studies report that consumer credit scores of small firm owners are used in lieu of firm characteristics \citep{BergerCowanFrame2011, DeYoungGlennonNigro2008, FrameSrnivasanWoosley2001}.} $c_{ij}$ is the $j$-th characteristic of the $i$-th firm, $j=1,\ldots,m$ and $i=1,\ldots,N$. We assume that $\lambda(\pmb{c})$ is derived from some credit scoring mechanism \citep{FrameSrnivasanWoosley2001,DeYoungGlennonNigro2008,Cole2016}. \cite[pp.~1259-1260]{Cressy1996} provides details on the variables considered by banks in this regard.  In order to determine the amount of loans to provide, the venture capitalist solves a program similar to the second-best semi-separating equilibrium in \eqref{subeq:FirmSort2}. At date 1 the states are simultaneously revealed to the firm and the lender, and beliefs about firm type are obtained from Bayesian updating from date 0. Good states occur with probability $\mu_g(e)$ where $e\in\{\underline{e},\overline{e}\}$ denotes low $(\underline{e})$ or high $(\overline{e})$ effort by the entrepreneur. We assume that bankruptcy is an exogenous phenomenon that occurs with probability $q$ and reflects the general failure rate of firms in the economy. Further, the optimal amount of loans $x\in\{\underline{x},\overline{x}\}$ determined at date 0 is distributed over the life of the contract. In that way, the principal gets to award the remainder of the loan after observing the performance of the firm in the first period. Suppose $x_0$ is provided at date 0 and the remainder $x-x_0$ is provided at date 1 according to the prevailing state. For instance, in the good state the firm gets the rest of the loan at date 1 with probability 1. However, in a bad state the lender provides the residue with probability $p_b$. If the firm goes bankrupt then it gets zero. This is the monitoring mechanism. At date 0 the firm wants to maximize the amount of loans it gets at date 1 under a no bankruptcy assumption. Therefore it exerts effort at unit cost such that the expected loan is maximized as follows:
\begin{align}
   e^\star\in\arg\max\{\mu_g(e)(x-x_0)+(1-\mu_g(e))p_b(x-x_0)-e\}
\end{align}
The firm participates in the program at date 1 if the net loan amount it receives is positive. Thus,
\begin{align}
   \mu_g(e^\star)(x-x_0)+(1-\mu_g(e^\star))p_b(x-x_0)-e^\star\geq 0\\
   x-x_0\geq\dfrac{e^\star}{\mu_g(e^\star)+(1-\mu_g(e^\star))p_b}>e^\star
\end{align}
Thus, the residual amount of loans provided at the end of the first period must be greater than the effort expended to obtain the loan. This result is at the core of one of the complaints about the effectiveness of benevolent programs. Firms complain that the high cost of applying for the loan does not make it worthwhile. However, the key to this result is that firms are provided with the residual amount of loans when bad states occur. \citet[pp.~26-27]{ArpingLoranthMorrison2010} reports that that is indeed the case for loans guaranteed by the Small Business Administration in the US between 1990 and 2000. Note that effort levels decline when the probability that the lender provides the loan in bad states increases.
\subsection{Welfare effects of program}\label{subsec:BorrowerMisclassified}
In this section we examine the welfare effects of error probabilities and strategic misrepresentation. That is, if the lender assigns an ex-ante probability to firm type based on the characteristic vector provided by the firm at date 0, then at date 1 after states are revealed the lender now has more information to update the initial probabilities. Error probabilities and strategic misrepresentation occurs when an ex-post outcome suggests that the firm was misclassified. From the public sector perspective, the inefficiencies associated with error probabilities result in welfare loss. At date 1 we assume that states are revealed through the performance of the firm in the prior period. Indicators such as sales and profits and procurement of new contracts send a signal to the lender that the firm is in a good state. In bad states sales and profits are low and a required threshold of new business may not have been procured. Under those scenarios the lender engages in Bayesian updating of risk classification. Assume that sales is a highly correlated signal (perhaps noisy) of states so that $s_G$ and $s_B$ are reported sales in good and bad states respectively. Then the corresponding posterior probabilities for high-risk and low-risk firms are given by:
\begin{align}
   \bar{\lambda}(\pmb{c}|\,s_B) &= \dfrac{(1-\mu_g(\bar{e}))\bar{\lambda}(\pmb{c})}{(1-\mu_g(\bar{e}))\bar{\lambda}(\pmb{c})+(1-\mu_g(\underline{e}))\underline{\lambda}(\pmb{c})}\label{PosteriorHighRisk}\\
   \underline{\lambda}(\pmb{c}|\,s_G) &= \dfrac{\mu_g(\underline{e})\underline{\lambda}(\pmb{c})}{\mu_g(\underline{e})\underline{\lambda}(\pmb{c})+\mu_g(\bar{e})\bar{\lambda}(\pmb{c})}\label{PosteriorLowRisk}
\end{align}
Error probabilities for high-risk and low-risk are given respectively by
\begin{align}
   \underline{\lambda}(\pmb{c}|\,s_B) = 1-\bar{\lambda}(\pmb{c}|\,s_B)\;\text{and}\;\bar{\lambda}(\pmb{c}|\,s_G)=1-\underline{\lambda}(\pmb{c}|\,s_G)
\end{align}
These probabilities reflect the event of a good \emph{ex-ante} decision being followed by a bad \emph{ex-post} outcome or state or a bad \emph{ex-ante} decision followed by a good state \emph{ex-post}. This brings us to the following lemma for reducing welfare loss.
\begin{lem}\label{lem:EqualEffort}
  A necessary condition for minimizing welfare loss from misclassification is that effort levels by low-risk firms and high-risk firms must be equal across states so that $\mu_g(\bar{e})=\mu_g(\underline{e})$.
\end{lem}
	Let $\check{x}=\bar{\lambda}(\pmb{c}|\,s_B)(\bar{x}-\bar{x}_0)+\underline{\lambda}(\pmb{c}|\,s_G)(\underline{x}-\underline{x}_0)$  be the weighted average residual loan used by the lender at date 1 to assign the cost of providing the loans. The expected welfare loss due to misclassification is represented by ``wrongful" surplus extraction given by:
\begin{align}
   W(\pmb{c}) &= \big(\bar{\theta}V(\bar{x}-\bar{x}_0)-C(\check{x})\big)\bar{\lambda}(\pmb{c}|\,s_G)+\big(\underline{\theta}V(\underline{x}-\underline{x}_0)-C(\check{x})\big)\underline{\lambda}(\pmb{c}|\,s_B)\label{eq:SurplusW}
\end{align}
The public policy objective is to \emph{minimize} $W(\pmb{c})$. This can be done only by minimizing $\bar{\lambda}(\pmb{c}|\,s_B)$  and  $\underline{\lambda}(\pmb{c}|\,s_B)$, which can only be minimized with respect to the characteristics vector. Hence the dependence of $W$ on $\pmb{c}$. The minimization problem depends on the specification of the model which is used to determine  $\lambda$.   We specify a logit model at date 0 so that:\footnote{The logit model predicts a probability of bankruptcy conditioned on borrower characteristics $\pmb{c}$. We chose logit because \citep{Lo1986} found that a logit model was more robust than the discriminant analysis approach to bankruptcy prediction popularized by \citet{Altman1968}. \citet{Cressy1996} used a probit model in his analysis. However, it is known that logit and probit models are quite similar.}
\begin{align}
   \operatorname{logit}(\lambda_{0,i}\mid\pmb{c}_0) &= \ln\Big(\frac{\lambda_{0,i}}{1-\lambda_{0,i}}\Big)=\pmb{c}^T_0\pmb{\beta}_0+\epsilon_{0,i}
\end{align}
where $\epsilon_{0,i}$ is an error term. At date 1 the model is updated to incorporate state variables $\pmb{c}_1$ that include sales and performance in the prior period, such as days sales outstanding, receivables, and days' receipts, so that
\begin{align}
   \operatorname{logit}(\lambda_{1,i}\mid\pmb{c}_1) &= \ln\Big(\frac{\lambda_{1,i}}{1-\lambda_{1,i}}\Big)=\pmb{c}^T_1\pmb{\beta}_1+\eta_{1,i}\label{eq:PosteriorOdds}
\end{align}
The value of $\pi$ obtained from estimates of \eqref{eq:PosteriorOdds} is minimized only to the extent that the model is correctly specified. Ideally, the principal would minimize all error probabilities simultaneously. For tractability, suppose that the relevant error probabilities are minimized at a common value.\footnote{This assumption is analogous to imposing symmetry between Type I and Type II classification errors. It allows the welfare-loss condition to be expressed in terms of a single posterior misclassification index.} Thus,
\begin{align}
   \bar{\lambda}_{\text{min}}(\pmb{c}|\,s_G)=\underline{\lambda}_{\text{min}}(\pmb{c}|\,s_B)=\lambda_{\text{min}}(\pmb{c}|\,s_G,s_B)\label{eq:LambdaEqual}
\end{align}
Dropping the ``min" and writing $\lambda(\pmb{c})$ in \eqref{eq:LambdaEqual} for notational convenience we find that \eqref{eq:SurplusW} is reduced to
\begin{align}
   W(\pmb{c}) &= \Big(A - 2C(\check{x})\Big)\lambda(\pmb{c})\label{eq:SurplusWmod},\;\text{where} \;A =\Big(\bar{\theta}V(\bar{x}-\bar{x}_0)+\underline{\theta}V(\underline{x}-\underline{x}_0)\Big)
\end{align}
Because misclassification can impose high program costs, the policy problem is to reduce $\lambda(\pmb{c})$ and therefore $W(\pmb{c})$. This result is reflected in the following
\begin{lem}\label{lem:MinimalWelfareLoss}
   The minimal welfare loss from program-related firm misclassification is $W^\star(\pmb{c})=A-2C^\star(\check{x}(\pmb{c}))=0$ where $C^\star(\check{x})=\frac{1}{2}\Big(\bar{\theta}V(\bar{x}-\bar{x}_0)+\underline{\theta}V(\underline{x}-\underline{x}_0)\Big)$
\end{lem}
\begin{prop}[Minimal welfare loss]\label{prop:MinimalWelfareLoss}
   The welfare loss is minimized when marginal cost is weighted by the marginal effect of firm characteristics on posterior classification probabilities, so that
\begin{equation}
\sum_i\beta_i\frac{\partial C(\check{x}(\pmb{c}))}{\partial c_i}=0.
\end{equation}
\end{prop}
Thus, minimal welfare loss is linked to posterior classification probabilities $\lambda$. This implies that date-1 audits can reduce welfare loss when they improve posterior classification. Economically, the condition says that the program should spend monitoring resources on characteristics only when those characteristics improve classification enough to justify their marginal cost. If a borrower characteristic is expensive to verify but has little effect on the posterior probability of misclassification, the weighted marginal-cost term is large relative to its informational value and the program should not rely heavily on it. Conversely, characteristics that sharply reduce posterior misclassification can justify higher monitoring or auditing costs. The welfare result therefore translates into an implementable audit rule: rank borrower characteristics by their marginal classification value per unit of monitoring cost.
\section{Aggregate Proof-of-Concept Evidence}\label{sec:ProofOfConcept}

The model is written at the level of borrower type, loan size, monitoring cost, and posterior classification. A direct empirical test would therefore require loan- or investment-level data containing borrower characteristics, ex ante classification, the amount and terms of each financing, monitoring intensity, subsequent performance, and any reclassification after new information is observed. In the United States, data with that level of granularity may be available for Small Business Investment Companies (SBICs) or Minority Enterprise Small Business Investment Companies (MESBICs).\footnote{An SBIC is a privately owned and managed investment firm that is licensed and regulated by the U.S. Small Business Administration (SBA) to provide critical equity capital, long-term loans, and hybrid financing to qualifying small businesses and startups.\\
MESBICs are government-chartered venture firms that can invest only in companies that are at least 51 percent owned by members of a minority group or persons recognized by the rules that govern MESBICs. They provide debt and equity capital to new, small independent businesses. Criteria for investment and size and type of investment vary from one firm to another.} In particular, the most useful data for testing the mechanism directly would be investment-level SBIC Form 1031 information and fund-level Form 468 information, linked to firm outcomes and delegated-monitor characteristics. Without access to that level of granularity, the empirical exercise in this section is deliberately limited. It asks whether publicly available aggregate data are at least consistent with the allocation margins emphasized by the theory.

The proof-of-concept exercise uses publicly available SBA SBIC State-by-State Financing Reports and SBIC Annual and Performance Reports \citep{SBAStateFinancing2026,SBAAnnualPerformance2026}. The state reports provide a state-year panel for fiscal years 2018--2025 with the number of financings, the number of businesses financed, and financing amounts. The annual reports provide program context and identify changes in program design, including the recent emphasis on program growth, investment priorities, and the Accrual Debenture structure. The Internet Appendix reports the sample construction, the 432 state-year observations used in the exercise, the variable dictionary, the concentration table, the figures, and the descriptive regression tables. Because the public state-year data do not identify individual borrowers, risk types, monitoring actions, or posterior beliefs, the exercise is descriptive rather than causal.

The aggregate evidence supports one central implication of the model: public-private capital is not allocated uniformly across the policy space. In fiscal year 2025, the five largest recipient states accounted for 42.8 percent of SBIC financing. The state-level Herfindahl-Hirschman index for financing amounts was 0.054 and the cross-state Gini coefficient was 0.61. These figures indicate substantial geographic concentration. That pattern is consistent with delegated allocation: private intermediaries appear to deploy publicly supported capital where deal flow, expected repayment, monitoring technology, and expected returns are favorable, rather than mechanically allocating funds in proportion to population or broad public-policy need.

The state-year evidence also speaks to the model's rationing intuition. Comparing each state's share of SBIC financing with its population share produces an allocation gap. States below parity receive less SBIC financing than their population share would predict. When state-level policy-target proxies are added, the descriptive exercise can ask whether states with weaker socioeconomic indicators receive systematically lower financing shares or lower financing per capita. Negative relationships would be consistent with the model's claim that delegated monitors may ration capital away from locations where screening and monitoring are more costly. However, because the public aggregate data do not identify the borrower pool, application demand, risk type, or monitoring cost, this evidence should be interpreted as suggestive only.

The public data are also useful for a limited program-design test. The annual reports describe changes in the structure and priorities of the SBIC program. A state-year panel can therefore examine whether financing volumes shift around observable program-design periods, such as the post-2023 period associated with the Accrual Debenture and broader program changes. Such tests can show whether aggregate financing volume responds to design margins, but they do not identify the mechanism by which program design affects borrower sorting, delegated monitoring, or welfare loss.

Thus, the public aggregate data support the model at the level of allocation patterns: SBIC capital is geographically concentrated, allocation gaps can be measured, and financing volumes can be compared across program-design periods. More granular data would be needed to test the paper's sharper theoretical predictions. First, borrower-level data are required to determine whether high-risk and low-risk firms receive the loan schedules implied by the separating contracts. Second, application-level data are required to distinguish credit rationing from weak demand. Third, monitoring and performance data are required to test whether posterior classification probabilities update as in the sequential model. Finally, fund-level and deal-level data are required to estimate whether welfare losses are minimized when marginal monitoring costs are weighted by the marginal effect of firm characteristics on classification probabilities. The proof-of-concept evidence therefore motivates the model and establishes empirical feasibility, but it does not substitute for a full micro-level test.

Because the public files do not contain the borrower-level fields needed for those sharper tests, the Internet Appendix also reports a simulation backstop. The simulation generates 60,000 borrower applications to 120 delegated SBIC/MESBIC-style funds over 36 states and 12 years. It is calibrated to the mechanisms in the model: delegated screening is noisy, target-area borrowers are harder to monitor, a policy-targeted intermediary channel can partially mitigate rationing, subsidies affect approval incentives, and monitoring affects default. The simulated regressions recover the model's comparative statics. Target-area borrowers have lower approval probabilities, the interaction between MESBIC status and target-area status is positive in the approval equation, screening signals and collateral raise approval and financing amounts, MESBIC financing is more debt intensive, and project quality reduces default. The exercise is not empirical evidence about historical MESBIC lending; it is a design-validation exercise showing that the proposed granular empirical tests have recoverable content if application-, monitoring-, and performance-level data become available.

\section{Conclusion}\label{sec:Conclusion}
This paper develops a mechanism-design model of delegated monitoring in public-private partnerships for access to credit. The main result is that agency costs can defeat benevolent public intent in two distinct ways. Imperfect screening and interest-rate sorting can generate underinvestment by excluding target firms that public policy is meant to support. Subsidized credit can generate overinvestment by making high-risk firms privately attractive to delegated monitors even when the induced allocation is socially inefficient.

The mechanism works through borrower misclassification and delegated-monitor incentives. Firms may gain by portraying themselves as target firms, while delegated monitors use imperfect application signals and credit-scoring rules to assign risk type. As agency costs rise, the feasible range of loans offered by the partnership contracts with expected risk. In the sequential extension, posterior updating disciplines the program only when firm characteristics and performance signals are sufficiently informative. Welfare loss from misclassification is minimized when the marginal cost of a firm characteristic is weighted by that characteristic's marginal effect on posterior classification probabilities.

The policy implication is that public-private credit programs should not be evaluated only by the volume of loans originated or the number of firms reached. The relevant design margin is whether the partnership induces additional lending to target firms without weakening screening, monitoring, and repayment incentives. Effective PPP design therefore requires explicit attention to delegated-monitor agency costs, risk-based subsidy rules, audit technology, and posterior performance evaluation.

The proof-of-concept evidence is consistent with this interpretation. In lieu of granular borrower-level SBIC or MESBIC data, aggregate SBIC state-year reports show geographic concentration, measurable allocation gaps, and variation across program-design periods. A simulation backstop then shows that the model's borrower-level predictions are recoverable in data with applications, delegated funds, borrower states, screening signals, contract type, loan size, and performance. These exercises support the model's allocation-pattern and empirical-design implications, while leaving historical borrower-level identification for future work with application, monitoring, and performance data.

With access to granular data, future empirical research in this area might investigate:
\begin{itemize}
	\item 	Do delegated public-private lenders allocate less capital to riskier target firms?
	\item 	Does public leverage induce overinvestment in weak borrowers?
	\item 	Are minority/target firms screened differently from non-target firms?
	\item 	Does monitoring intensity or loan size vary with observable risk?
	\item 	Do delegated intermediaries ration credit geographically or by industry despite public-policy objectives?
\end{itemize}

\singlespacing
\bibliographystyle{chicago}        
\section*{}
\addcontentsline{toc}{section}{References} 
\bibliography{PubPrivPartner_arXiv}         

\appendix
\setcounter{section}{0}
\setcounter{equation}{0}
\renewcommand{\thesection}{A}
\renewcommand{\theequation}{A.\arabic{equation}}
\renewcommand{\theHequation}{A.\arabic{equation}}
\renewcommand{\theHsection}{A}
\section{Proofs}\label{app:proofs}

\subsection{Proof of Proposition~\ref{prop:LoanFirstOrderCondition}}
\begin{proof}
Because firm surplus is type separable, the continuous second-best objective can be written as
\begin{equation}
   \Phi(x)=\int_{\overline{x}}^{\underline{x}}\int_{\overline{\theta}}^{\underline{\theta}}
   \big[\theta V(x)-\lambda(\theta)C(x)\big]dF(x)dG(\theta).
\end{equation}
For an interior optimum $\tilde{x}\in(\overline{x},\underline{x})$, differentiation under the integral sign gives
\begin{equation}
   \int_{\overline{\theta}}^{\underline{\theta}}
   \big[\theta V_x(\tilde{x})-\lambda(\theta)C_x(\tilde{x})\big]dG(\theta)=0.
\end{equation}
Continuity of the integrand and the integral mean-value theorem imply that some representative type $\theta^*\in[\overline{\theta},\underline{\theta}]$ satisfies
\begin{equation}
   \theta^*V_x(\tilde{x})-\lambda(\theta^*)C_x(\tilde{x})=0.
\end{equation}
The monotonicity condition on $G$ selects the modal classification type, $\theta^*=\arg\max_\theta g(\theta)$. Hence the optimal interior loan belongs to
\begin{equation}
   \bigl\{x: \theta^*V_x-\lambda(\theta^*)C_x=0,\;\theta^*=\arg\max_\theta g(\theta)\bigr\}\cap[\overline{x},\underline{x}],
\end{equation}
which proves the claim.
\end{proof}

\subsection{Proof of Corollary~\ref{cor:LoanOptimalityMC_Condition}}
\begin{proof}
The first-order condition from Proposition~\ref{prop:LoanFirstOrderCondition} is
\begin{equation}
   \theta^*V_x(\tilde{x})=\lambda(\theta^*)C_x(\tilde{x}).
\end{equation}
The left-hand side is the marginal utility of an additional unit of lending for the representative type. The right-hand side is marginal monitoring cost weighted by the probability that the loan is assigned to that type. Thus marginal utility equals expected marginal cost.
\end{proof}

\subsection{Proof of Proposition~\ref{prop:StrategicRange}}
\begin{proof}
In the increasing-cost region, the monitoring condition in the main text requires
\begin{equation}
   C^\prime(\widetilde{x})d\overline{\lambda}+\overline{\lambda}C^{\prime\prime}(\widetilde{x})(\overline{x}-\underline{x})d\overline{\lambda}\geq 0.
\end{equation}
For $d\overline{\lambda}>0$, this is equivalent to
\begin{equation}
   C^\prime(\widetilde{x})\geq \overline{\lambda}C^{\prime\prime}(\widetilde{x})(\underline{x}-\overline{x}).
\end{equation}
Solving for the range of loans gives
\begin{equation}
   (\underline{x}-\overline{x})\leq \Big(\overline{\lambda}\frac{C^{\prime\prime}(\widetilde{x})}{C^\prime(\widetilde{x})}\Big)^{-1}.
\end{equation}
\end{proof}

\subsection{Proof of Lemma~\ref{lem:EqualEffort}}
\begin{proof}
If the posterior classification system is state-invariant at the optimum, then the posterior probability assigned to a high-risk firm after a good state must equal the posterior probability assigned after a bad state. Using Bayes' rule, this requires
\begin{equation}
  \frac{(1-\mu_g(\bar{e}))\bar{\lambda}(\pmb{c})}{(1-\mu_g(\bar{e}))\bar{\lambda}(\pmb{c})+(1-\mu_g(\underline{e}))\underline{\lambda}(\pmb{c})}
  =
  \frac{\mu_g(\bar{e})\bar{\lambda}(\pmb{c})}{\mu_g(\underline{e})\underline{\lambda}(\pmb{c})+\mu_g(\bar{e})\bar{\lambda}(\pmb{c})}.
\end{equation}
Cross-multiplication and cancellation of common positive prior terms imply $\mu_g(\bar{e})=\mu_g(\underline{e})$. Hence equal effort across states is necessary for minimizing welfare loss under this posterior-equality criterion.
\end{proof}

\subsection{Proof of Lemma~\ref{lem:MinimalWelfareLoss}}
\begin{proof}
Since $\lambda(\pmb{c})$ is a positive posterior misclassification probability, $W(\pmb{c})=0$ if and only if the bracketed term is zero. Hence
\begin{equation}
   A-2C^*(\check{x}(\pmb{c}))=0.
\end{equation}
Substituting the definition of $A$ gives
\begin{equation}
   C^*(\check{x}(\pmb{c}))=\frac{1}{2}\Big(\overline{\theta}V(\overline{x}-\overline{x}_0)+\underline{\theta}V(\underline{x}-\underline{x}_0)\Big).
\end{equation}
This is the minimal attainable welfare-loss benchmark.
\end{proof}

\subsection{Proof of Proposition~\ref{prop:MinimalWelfareLoss}}
\begin{proof}
Let
\begin{equation}
   \lambda(\pmb{c})=\frac{\exp(\pmb{c}^T\pmb{\beta})}{1+\exp(\pmb{c}^T\pmb{\beta})}.
\end{equation}
Then
\begin{equation}
   \frac{\partial\lambda(\pmb{c})}{\partial c_i}=\lambda(\pmb{c})\big(1-\lambda(\pmb{c})\big)\beta_i.
\end{equation}
Using $W(\pmb{c})=\big(A-2C(\check{x}(\pmb{c}))\big)\lambda(\pmb{c})$, the directional derivative of $W$ along the logit-index direction $d\pmb{c}=\pmb{\beta}\,dt$ is
\begin{equation}
   \frac{dW}{dt}
   =-2\lambda(\pmb{c})\sum_i \beta_i\frac{\partial C(\check{x}(\pmb{c}))}{\partial c_i}
   +\big(A-2C(\check{x}(\pmb{c}))\big)
      \lambda(\pmb{c})\big(1-\lambda(\pmb{c})\big)\sum_i\beta_i^2.
\end{equation}
At the zero-loss benchmark in Lemma~\ref{lem:MinimalWelfareLoss}, the second term vanishes. Since $\lambda(\pmb{c})>0$, the necessary first-order condition $dW/dt=0$ implies
\begin{equation}
   \sum_i\beta_i\frac{\partial C(\check{x}(\pmb{c}))}{\partial c_i}=0.
\end{equation}
\end{proof}

\clearpage
\phantomsection
\addcontentsline{toc}{section}{Internet Appendix}
\includepdf[pages=-,pagecommand={\thispagestyle{plain}}]{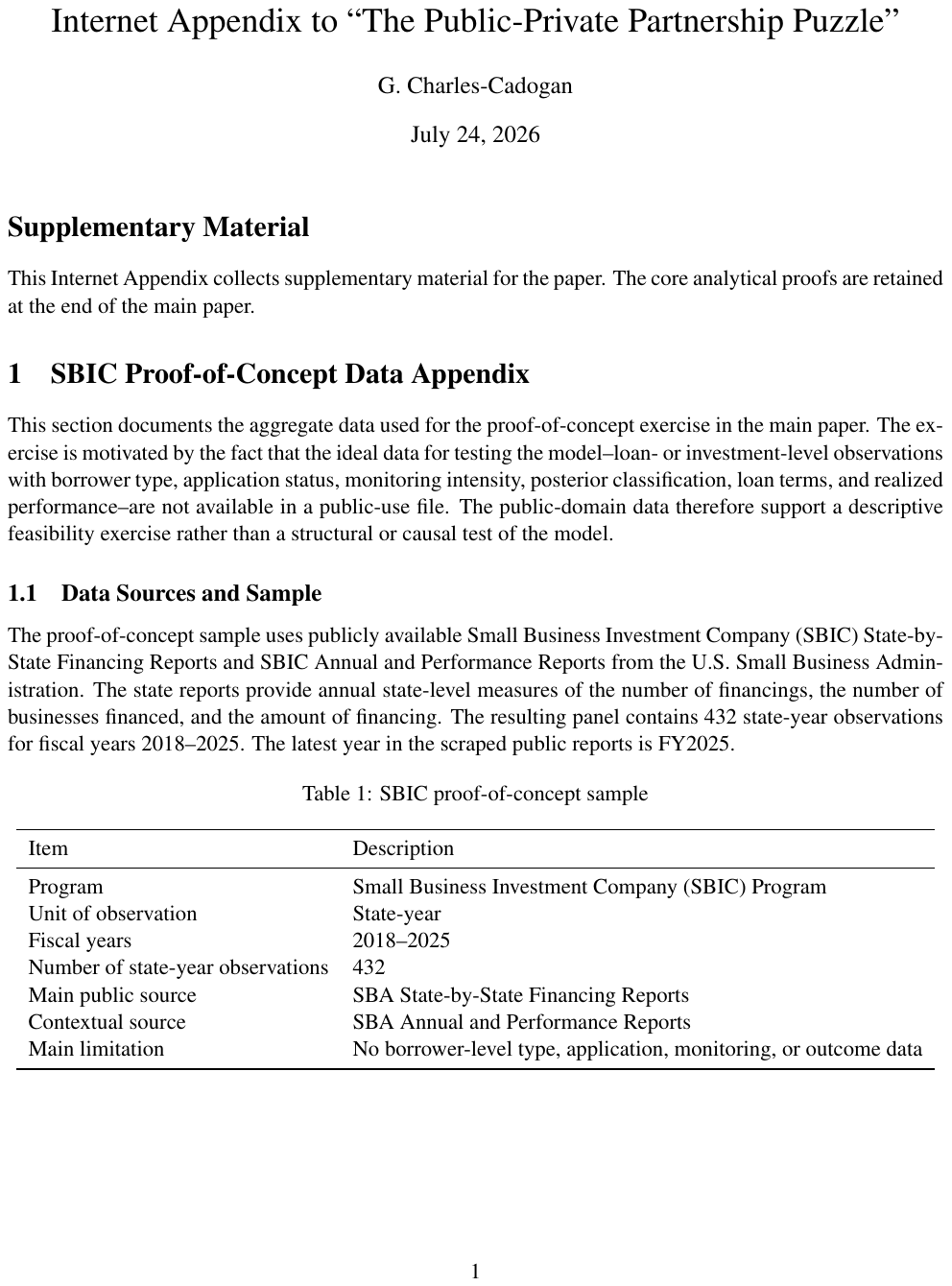}
\end{document}